\documentclass[11pt]{amsart}
 \usepackage[text={6.5in,9in},centering]{geometry}
\usepackage{graphicx}
\usepackage{times}
\usepackage{amssymb}
\usepackage{amsmath}
\usepackage{multicol}
\usepackage{epstopdf}
\usepackage{latexsym}
\usepackage{algorithm}
\usepackage{algorithmic}
\usepackage{fullpage}
\usepackage[foot]{amsaddr}
\usepackage{tikz}

\usepackage{verbatim}
\usetikzlibrary{arrows,shapes}

\newcounter{foo}
\newtheorem{theorem}[foo]{Theorem}
\newtheorem{lemma}{Lemma}[section]

\newtheorem{corollary}[lemma]{Corollary}
\newtheorem{observation}[lemma]{Observation}

\newtheorem{definition}{Definition}[section]
\newtheorem{proposition}[foo]{Proposition}

\usepackage{url}

\def\calD{{\mathcal{D}}}

\def\calM{{\mathcal{M}}}

\def\calP{{\mathcal{P}}}

\def\calX{{\mathcal{X}}}

\newcommand{\reals}{\mathbb{R}}

\newcommand{\prob}[1]{\textsc{#1}}
\newcommand{\Capacity}{\prob{Capacity}}
\newcommand{\capacity}{\Capacity}
\newcommand{\Scheduling}{\prob{Scheduling}}
\newcommand{\scheduling}{\Scheduling}

\newcommand{\geomodel}{\textsc{geo-SINR}}

\newcommand{\define}[1]{\emph{#1}} 
\newcommand{\degree}{^\circ}

\newcommand{\mypara}[1]{\smallskip\noindent\textbf{#1.}}  
\newcommand{\tightpara}[1]{\noindent\textbf{#1.}}  
\newcommand{\inddim}{D}  

\newcommand{\authorcomment}[1]{}

\urldef{\ourmail}\url{{mbodlaender,magnusmh}@gmail.com}

\title{Beyond Geometry : Towards Fully Realistic Wireless Models
\thanks{ICE-TCS, School of Computer Science, Reykjavik University, 101 Reykjavik, Iceland. \url{
\{mbodlaender,magnusmh\}@gmail.com}. tel: +354-825-6384.
Supported by Icelandic Research Fund grant-of-excellence no.~120032011.}}
\author{Marijke H.L. Bodlaender \and Magn\'us M. Halld\'orsson}

\begin{document}

\begin{titlepage}

\date{\today}
\thispagestyle{empty}

\begin{abstract}
Signal-strength models of wireless communications capture the gradual
fading of signals and the additivity of interference. As such, they
are closer to reality than other models. However, nearly all theoretic
work in the SINR model depends on the assumption of smooth geometric
decay, one that is true in free space but is far off in actual
environments. The challenge is to model realistic environments,
including walls, obstacles, reflections and anisotropic antennas,
without making the models algorithmically impractical or analytically
intractable.

We present a simple solution that allows the modeling of arbitrary
static situations by moving from geometry to arbitrary \emph{decay
  spaces}.  The complexity of a setting is captured by a
\emph{metricity} parameter $\zeta$ that indicates how far the decay
space is from satisfying the triangular inequality.  All results that
hold in the SINR model in general metrics carry over to decay spaces,
with the resulting time complexity and approximation depending on
$\zeta$ in the same way that the original results depends on the path
loss term $\alpha$.  For distributed algorithms, that to date have
appeared to necessarily depend on the planarity, we indicate how they
can be adapted to arbitrary decay spaces at a cost in time complexity
that depends on a \emph{fading} parameter of the decay space. In
particular, for decay spaces that are \emph{doubling}, the parameter
is constant-bounded.

Finally, we explore the dependence on $\zeta$ in the approximability
of core problems. In particular, we observe that the capacity maximization problem 
has exponential upper and lower bounds in terms of $\zeta$ in general
decay spaces. In Euclidean metrics and related growth-bounded decay
spaces, the performance depends on the exact metricity definition, 
with a polynomial upper bound in terms of $\zeta$, but an exponential
lower bound in terms of a variant parameter $\phi$.
On the plane, the upper bound result actually yields the first approximation of 
a capacity-type SINR problem that is subexponential in $\alpha$.
\end{abstract}

\maketitle

\vspace{4em}


\end{titlepage}

\section{Introduction}
\label{sec:intro}

Signal-strength models of wireless communications capture the gradual
fading of signals and the additivity of interference. As such, they
are closer to reality than other models. In spite of the apparent
great complexity of such models, various fundamental problems have
been resolved analytically in recent years. These also seem essential for
studying certain properties of wireless networks, such as capacity
\cite{kumar00}, or connectivity and aggregation, which can be achieved
in logarithmic rounds in worst case \cite{MoWa06,SODA12}.

Nearly all theoretic work in signal-strength models have been done in
the ``SINR model'' that assumes that signals decay as a smooth
polynomial function of distance. We shall refer to this as the
{\geomodel} model. This assumption about decay (or \emph{path loss})
is true in free space, but turns out to be far off in actual
environments, as shown by a long history of experimental studies
(e.g., \cite{kotz2004experimental}). Quoting a recent meta-study,
\cite{baccour2012radio}, ``link quality is not correlated with distance.'' 
Experimental studies have long ago jettisoned the geometric path loss
assumption. 
This questions the wisdom of studying ``SINR models'' analytically, given the added
effort and complexity.

One hope might be that results in the ``basic SINR model'' could
eventually carry some insights that would be of use in more detailed
models that capture more of reality. Yet, there are no proposed
intermediate models, and real environments consist of 
assortments of walls, ceilings and obstacles,
as well as complex interactions involving reflections, shadowing, 
multi-path signals, and anisotropic (or even directional) antennas.
It might seem near impossible to capture this all without 
making the resulting models hopelessly impractical for algorithm
design and/or analytically intractable. 

\tightpara{Our contributions}
We present a simple solution that allows the modeling of arbitrary static
situations by moving from geometry to arbitrary \emph{decay spaces}.
The decay between two ordered nodes is the reduction in the strength
of a signal sent from the first node to the second.
By signal-strength measurements, that almost any cheap node can
perform today, these decays capture the \emph{truth on the ground}.
The complexity of a setting is captured by a
\emph{metricity} parameter $\zeta$ that indicates how far the decay
space is from satisfying the triangular inequality.

\emph{All} results that hold in the SINR model in general metrics carry over
to decay spaces, with the resulting time complexity and approximation
depending on $\zeta$ in the same way that the original results depends
on the path loss term $\alpha$.

For distributed algorithms, that to date have appeared to necessarily depend on
the planarity, we introduce a \emph{fading} parameter of the decay space
and indicate they can be adapted to arbitrary decay spaces
at a cost in time complexity that depends on a \emph{fading} parameter
of the decay space. In particular, for decay spaces that are
\emph{doubling}, the parameter is constant-bounded.

Finally, we explore the dependence on $\zeta$ in the approximability
of core problems. In particular, we observe that the {\capacity} problem 
has exponential upper and lower bounds in terms of $\zeta$ in general
decay spaces. In Euclidean metrics and related growth-bounded decay
spaces, the performance depends on the exact metricity definition, 
with a polynomial upper bound in terms of $\zeta$, but an exponential
lower bound in terms of a variant parameter $\phi$.

One may ask if we are being led to yet another model that will later been shown unrealistic. Fortunately, numerous experimental studies have verified the remaining key assumptions in wide range of situations and technology \cite{son2006,MaheshwariJD2008,chen2010,sevani2012sir,us:ICDCS14}: additivity of interference, SINR capture effectiveness (the near-thresholding relationship between SINR level and packet reception rate), and invariability of wireless conditions in static environments. Thus, we may finally be reaching a wireless model that is a close approximation of reality, yet usable algorithmically and analytically.
That said, one should not discount the value of abstractions or the potentially value of simple models. Also, modeling dynamic and mobile situations, which is outside the scope of our work, remains a highly important (and largely open) issue.

\tightpara{Related work}
%
The ``abstract SINR'' model captures, like decay spaces, arbitrary pairwise path loss. Some positive results hold in that model, e.g., 
distributed power assignment of feasible sets \cite{LotkerPPP11},
reductions involving Rayleigh fading \cite{Dams2012}, and 
special cases of capacity maximization \cite{us:algosensors11}.
However, for most problems of interest, extremely strong inapproximability results hold \cite{GHWW09,khot2006better}. 
Thus, it is essential to use near-metric properties of the decay space.

The introduction of general metrics (apparently first in
\cite{FKV09,FKRV09}) was a significant step in extending SINR theory
beyond geometric assumptions.
Fading metrics \cite{us:talg12} were identified to capture the main
property required from the planar setting.
The concept of \emph{inductive independence} \cite{KV10,hoeferspaa}
has heralded a more systematic approach to SINR analysis, and can by
itself be seen as parameter of the decay space. 
Same holds for \emph{$C$-independence} \cite{infocom11,dams2013sleeping}
in the case of uniform power.

In a sibling paper \cite{us:ICDCS14}, we introduced decay spaces and
metricity with a focus on experimental validation.
The experimental results align with previous results (e.g., 
\cite{son2006,MaheshwariJD2008,chen2010,sevani2012sir}) that
whereas geometric decay is far off, other factors of the ``SINR
model'' closely approximate reality.
In the current paper, for comparison, we substantiate our claims 
of theory transfer, 
treat the fading necessary for distributed algorithms, give lower
bound results in terms of metricity parameters, and show that capacity
approximation in the plane depends only polynomially on the path loss 
term $\alpha$.

\mypara{Outline of the rest of the paper}
In the next section, we introduce decay spaces (formal definitions, the metricity parameter and how these spaces can be populated),
and indicate how previous results in metric spaces carry over.
In Sec.~\ref{sec:fading}, we address the core requirement of \emph{fading} for distributed algorithms, introduce a parameter that extends their reach to arbitrary spaces, and prove constant upper bounds in spaces with bounded doubling dimension.
The impact of metricity parameters on approximability is treated in Sec.~\ref{sec:dependence}.

\section{Decay Spaces}
\label{sec:models}

\subsection{Signal-strength models}
\label{sec:ss}

The \emph{abstract SINR} model has two key properties:
\textbf{(i)} signal decays as it travels from a sender to a receiver,
and \textbf{(ii)} interference -- signals from other than the intended transmitter -- accumulates.
Transmission succeeds if and only if the interference is below a given threshold.

Formally, a \emph{link} $l_v = (s_v, r_v)$ is given by a pair of
nodes, sender $s_v$ and a receiver $r_v$.  The \emph{channel gain} $G_{uv}$ denotes the multiplicative
decay in the signal of $l_u$ as received at $r_v$.  The
\emph{interference} $I_{uv}$ of sender $s_{u}$ (of link $l_u$) on 
the receiver $r_v$ (of link $l_v$) is $P_u G_{uv}$,
where $P_v$ is the power used by $s_v$.
When $u=v$, we refer to $I_{vv}$ as the \emph{signal strength} of link $l_v$.
If a set $S$ of
links transmits simultaneously, then the \emph{signal to noise and
  interference ratio} (SINR) at $l_v$ is
\begin{equation}
 \text{SINR}_v := \frac{I_{vv}}{N + \sum_{u \in S} I_{uv}} = 
   \frac{P_v G_{vv}}{N + \sum_{u \in S} P_v G_{uv}}\ ,
\label{eqn:sinr}
\end{equation}
where $N$ is the ambient noise.

We refer to the standard signal-strength model as the {\geomodel} model, which 
adds to the SINR formula the assumption of \emph{geometric path loss}:
that signal decays proportional to a fixed polynomial of the distance, 
\emph{i.e.}, 
$G_{uv} = d(s_u, r_v)^{-\alpha}$,
 where the \emph{path loss term} $\alpha$ is assumed to be an
 arbitrary but fixed constant between 1 and 6.  
This assumption is valid in free space, with $\alpha=2$ in perfect vacuum.  

The last assumption made in theoretical models is \emph{thresholding}:
the transmission of $l_v$ is \emph{successful} iff $\text{SINR}_v
\ge \beta$, where $\beta \ge 1$ is a hardware-dependent constant.  We shall
also make this assumption. It's been shown by Dams, Kesselheim and
Hoefer \cite{Dams2012} that certain models that include a randomized
filter in this decision can be efficiently simulated by thresholding
algorithms.

\subsection{Metrics and Decay Spaces}
\label{sec:decayspaces}

We seek to model arbitrary path loss that is independent of distance.
We capture this by a \emph{decay} function $f$ of pairs of points (or
nodes) so that $G_{uv} = 1/f(s_u, r_v)$.

We shall formulate signal decay as \emph{decay spaces}.  Decays
between distinct points are always positive.  Exactly what happens at
a given point (i.e., the value of $f(p,p)$) is immaterial to our
consideration, since we may assume that all nodes are distinct.  


\begin{definition}
  A \emph{decay space} is a pair $\calD = (V,f)$, where $V$ is a
  discrete set of nodes (or points) and $f$ is a mapping (or matrix) $f:V\times V
  \rightarrow \reals_{\ge 0}$ that associates values (\emph{decays}) with
  ordered pairs of nodes.  The decays satisfy: i) $f(p,q) \ge 0$
  (non-negativity), and ii) $f(p,q)=0$ if and only if $p=q$ (the
  identity of indiscernibles).
\end{definition}

Decay spaces need not be symmetric nor obey the triangular inequality.
Such spaces are known as \emph{pre-metrics}.
As shorthand, we write $f_{pq} = f(p,q)$.


Decay space can either represent the truth-on-the-ground, or
its representation/approximation as data. They are relatively easily obtained by measurements, which even the cheapest gadgets today provide. 
They can also be inferred by packet reception rates, or predicted by heuristic  or environmental models \cite{Goldsmith}.

\subsubsection*{Metricity}
We introduced in \cite{us:ICDCS14} a parameter that represents how close the decay matrix is to a distance metric. 
\begin{definition}
The \emph{metricity} $\zeta(\calD)$ of a decay space $\calD = (V,f)$ is
the smallest number such that, for every triplet $x,y,z \in V$,
  \begin{equation}
     f(x,y)^{1/\zeta} \le f(x,z)^{1/\zeta} + f(z,y)^{1/\zeta}\ .
  \label{eq:zeta}
  \end{equation}
\end{definition}
Note that $\zeta$ is well-defined since $\zeta_0 = \lg (\max_{x,y}
f(x,y))/(\min_{x,y} f(x,y))$ satisfies (\ref{eq:zeta}).
In the case of geometric path loss, $\zeta = \alpha$, since
$f(x,y) = d(x,y)^\alpha$. 

We define \emph{quasi-distances} between nodes in a decay space 
by $d(p,q) = f_{pq}^{1/\zeta}$. Let $d_{pq}=d(p,q)$ for short.
These quasi-distances induce a \emph{quasi-metric} 
$\calD' = (V,d)$, i.e., a metric except for the possible lack of symmetry.
In the Euclidean setting, quasi-distances are simply the Euclidean distances.


\subsection{Theory transfer}
\label{sec:theorytransfer}

The lion share of the theoretic literature on signal-strength models
can be converted to decay spaces with limited effort.
We aim here to clarify and substantiate that observation.
Our objective is for the non-specialist to be able to determine with
limited effort which results do hold in the decay model and which don't and additionally, when the question arises, which properties of metric
and/or decay spaces are necessary for correct functioning.

In this section, we focus on what is needed for results to hold in
arbitrary decay spaces. In the following section, we deal with results that require special space properties, particularly in the context of distributed algorithms.
By a \emph{result}, we mean a combination of an algorithm or a protocol and its
analysis.

The complexity of a result can be a function of the metric/space.
Here, complexity refers to measures like time and message count, 
but also performance
measures like approximability. In particular, these measures have
nearly always been functions of the metric parameters, such as the
path loss term $\alpha$, but this dependence is often hidden in big-oh notation.

We make the following sweeping assertion (stated without substantiation
in the sibling paper \cite{us:ICDCS14}):

\begin{proposition}
  If a {\geomodel} result only requires metric properties (symmetry, triangular
  inequality), then it holds equally well in arbitrary decay spaces.  Symmetry
  is required of the decay space only if it was required in the
  original setting.  The relevant complexity measure (time,
  approximation) grows with $\zeta$ in the same manner as for the
  original result in terms of $\alpha$.
\label{prop:metric}
\end{proposition}

\begin{proof}
The quasi-distances $d$ of a decay space $\calD=(V,f)$
form a quasi-metric $\calD'=(V,d)$, which becomes a metric 
iff $\calD$ satisfies symmetry. 
Applying the original result to the metric $\calD'$ with path
loss constant $\zeta(\calD)$ gives an equivalent solution to the problem
on the decay space $\calD$.
\end{proof}

Specifically, the following results on the following problems carry over without change: 
capacity maximization \cite{SODA11,KesselheimESA12}, 
scheduling \cite{FKRV09,FKV09},
weighted capacity \cite{us:talg12,us:Infocom12}, 
spectrum auctions \cite{hoeferspaa,HoeferK12},
relationship between power control regimes \cite{tonoyan2011a,us:SODA13},
dynamic packet scheduling \cite{CISS12,sirocco12,kesselheimStability,us:SODA13}, 
distributed scheduling \cite{KV10,icalp11}, and
distributed capacity maximization with regret-minimization \cite{infocom11}
(extended for jamming \cite{dams2013jamming}, 
online requests with stochastic assumptions \cite{GHKSV13},
and changing spectrum availability \cite{dams2013sleeping}).

We can also make an immediate observation regarding methods that
hold for restricted metrics.
\begin{observation}
If a result holds in {\geomodel} for a given class $\calM$ of metrics,
then it holds equally in those decay spaces whose induced quasi-metric
is contained in $\calM$.
\label{asst:metric}
\end{observation}

\mypara{Results that do not carry over to decay spaces}
%
There remains a large amount of work in {\geomodel} that depends on
\emph{positions} (or distributions thereof). 
Such results are necessarily tied to geometry, although with some work
it may be possible to extend them to other decay spaces.

A common use of positional information is by partitioning the
plane, so as to make simultaneous communication non-conflicting. This
is particularly an issue for deterministic distributed algorithms.
Examples of this include deterministic distributed broadcast
\cite{JurdzinskiKRS13, JurdzinskiKS13FCT} and local broadcast
\cite{JurdzinskiK12, FuchsW13}.  Also, some centralized approximation
algorithms and heuristics for {\capacity} and {\scheduling} of
\cite{gouss2007, DBLP:journals/corr/abs-1208-0627}.  
Occasionally, angles are used, e.g.\cite{GHWW09}, which does not carry over
(but see Sec.~\ref{sec:polya}).

There is also a large literature on average case analysis, typically
assuming a uniform distribution of points in the plane, starting with
an influential paper of Gupta and Kumar \cite{kumar00} that first
introduced {\geomodel}.

Finally, SINR diagrams \cite{AvinEKLPR12} (and follow-up work of
subsets of the authors) uses intrinsically topological properties of Euclidean
metrics.

\subsection{Additional definitions: Power, affectance, separability}

We will work with a total order $\prec$ on the links, where $l_v \prec l_w$ implies that $f_{vv} \le f_{ww}$. 
%
A power assignment $\calP$ is \emph{monotone} if both $P_v \le P_w$
and $\frac{P_w}{f_{ww}} \le \frac{P_v}{f_{vv}}$ hold whenever $l_v \prec l_w$.
\footnote{This corresponds to \emph{length monotone} and
  \emph{sub-linear} power assignments in {\geomodel}.} This captures
the main power strategies, including uniform and linear power.

We modify the notion of \emph{affectance} \cite{GHWW09,HW09,KV10}:
The affectance $a^{\calP}_w(v)$ of link $l_w$ on link $l_v$ under power assignment $\calP$  is the interference of $l_w$ on $l_v$ normalized to the signal strength (power received) of $l_v$, or
\[
a_w(v) = \min \left(1, c_v \frac{P_w G_{wv}}{P_v G_{vv}}\right) = \min \left(1, c_v \frac{P_w}{P_v} \frac{f_{vv}}{f_{wv}}\right)\ ,
\]
where $c_v= \frac{\beta}{1-\beta N/(P_v G_{vv})} > \beta$ is a constant depending only on universal constants and the signal strength $G_{vv}$ of $l_v$, indicating the extent to which the ambient noise affects the transmission. 
We drop $\calP$ when clear from context.
Furthermore let $a_v(v) = 0$. For a set $S$ of links and link $l_v$, let $a_v(S) = \sum_{l_w \in S} a_v(w)$ be the \emph{out-affectance} of $v$ on $S$ and $a_S(v) = \sum_{l_w \in S} a_w(v)$ be the \emph{in-affectance}.
Assuming $S$ contains at least two links we can rewrite Eqn.~\ref{eqn:sinr} as $a_S(v) \leq 1$ and this is the form we will use.
A set $S$ of links is \emph{feasible} if $a_S(v) \leq 1$ 
and more generally \emph{$K$-feasible} if $a_v(S) \le 1/K$.



Define $d_{vw} = d(l_v,l_w) = \min(d(s_v,r_w),d(s_w,r_v), d(s_v,s_w),d(r_v,r_w))$ as the (quasi-)distance between two links $l_v$ and $l_w$.
Let $d_{vv} = d(s_v,r_v)$.
A link $l_v$ is said to be \emph{$\eta$-separated} from a set $L$ of links, for
parameter $\eta$, if $d(l_v,l_w) \ge \eta d_{vv}$ for every $l_w \in
L$.  A set $L$ is $\eta$-separated if each link in $L$ is
$\eta$-separated from the rest of the set.

Let $e$ refer to the base of the natural logarithm and recall that $1+x \le e^x$, for any value $x$.


\section{Fading Properties and Distributed Algorithms}
\label{sec:fading}

In the study of distributed algorithms in $\geomodel$ in the plane, the
standard assumption is that the path loss constant $\alpha$ is
strictly larger than 2. The reason for this is that when $\alpha > 2$, nodes
that are spatially well separated will not affect each other by too
much, a property that does not hold when $\alpha \le 2$.
This property is generalized to doubling metrics whose doubling dimension is strictly smaller than the path loss constant $\alpha$, dubbed \emph{fading metrics} \cite{us:talg12}.
We call this property, that the sum of affectances from spatially separated transmitting nodes 
converges, the \emph{fading} property. For the most common type of distributed algorithm to work, this has to be bounded.

We define a parameter $\gamma$ that captures the fading effect.
Let $\calX(r)$ be the space of all $r$-separated subsets in $V$.
\begin{definition}
The \emph{fading value} $\gamma_z(r)$ of a node $z$
relative to a separation term $r$ is 
\[ \gamma_z(r) = r \max_{X \in \calX(r)} \sum_{x \in X} 1/f_{xz}\ . \]
The \emph{fading parameter} $\gamma$ of a decay space 
is the maximum fading value of a node in the space,
$\gamma = \gamma(r) = \max_{z \in V} \gamma_z(r)$,
relative to a given separation term $r$.
\label{defn:fading}
\end{definition}


That is, the total interference $I_S(z)$ experienced by a node $z$ from an $r$-separated set $S$ (of senders) using uniform power $P$ is at most $\gamma(r) \cdot P/r$. Thus, if the intended signal comes from an $r$-neighborhood (in decay space), then the resulting affectance is bounded by $a_S(z) \le \frac{\gamma(r) P/r}{P/r} = \gamma(r)$.

Until now, $\gamma$ has been expected to be an absolute constant.
However, we can now simply treat it as a parameter and thus
handle arbitrary decay spaces by distributed algorithms.
Thus, we can achieve significantly more generality than before.
This would necessarily come at the cost of extra time complexity.

\subsection{Fading spaces}

We identify a large class of decay spaces for which the fading parameter is small.
These are generalizations of fading metrics.

First, some additional notation.  The \emph{$t$-ball} $B(y,t) = \{x
\in V | f(x,y) < t\}$ centered at $y$ with radius $t$
contains all points $x$ for which decay to $y$ is less than $t$.  A
set $Y \subseteq V$ is a \emph{$t$-packing} if $f(x,y) > 2t$, for any
$x,y \in V$.  Thus, $Y$ is a $t$-packing iff the set
$\{B(y,t)\}_{y \in Y}$ of balls are disjoint.  The \emph{$t$-packing
  number} $\mathcal{P}(\mathcal{B}, t)$ is the size of the largest
$t$-packing into the body $\mathcal{B}$.

Intuitively, a space is \emph{doubling} if the number of mutually unit-separated points 
within a given distance from a center increases by at most a polynomial of the distance.

\begin{definition}
Let $\calD = (V,f)$ be a decay space.
Define $g_\calD(q) = \max_{x \in V} \max_{r \in \reals^+}
\calP(B(x,r),r/q)$, as the size of the densest $q$-packing in $\calD$.
The \define{Assouad dimension $A$ of $\calD$ with parameter $C$} is
given by
  \[ A(\calD) = \max_q \log_q \left( \frac{g(q)}{C} \right)\ . \]
\end{definition}
$A(\calD)$ is in effect the minimum degree $k$ for which sizes of
$t$-packings can be bounded by $O(t^k)$, for all $t$.
Note that that $A(\mathbb{R}^k) = k$ \cite{Heinonen}.

\begin{definition}
A \emph{fading space} is a decay space $\calD$ with Assouad dimension
strictly smaller than 1, $A(\calD) < 1$, w.r.t.\ some absolute constant $C$.
\end{definition}

\subsection{Annulus argument}
Most randomized algorithms (e.g. in \cite{PODC13} and \cite{Yu12}) ensure that in any given neighborhood
(defined as the set of nodes to which a given node can communicate directly),
the expected number of transmissions in a slot is bounded above by a certain constant.
This ensures that the total expected affectance from other nodes transmitting
is also bounded by a (different) constant.
By adjusting the constants appropriately, one can focus only on the local behavior.
Some deterministic algorithms similarly ensure a spatial separation of sending (and thus possibly interfering) nodes and use this property to bound the total affectance from these nodes.

All proofs of the discussed sort use a common approach.  They define
some type of separation between interfering nodes which can be a
(probabilistic) constant density, a hard minimum distance between
nodes or links or similar.  Then the interference at a node $v$ is
bounded, either directly or, if the node is receiver of a
predefined link, as the (possibly probabilistic) affectance on the
node.  To do this we draw concentric circles around $v$, cutting the
space around $v$ up into annuli.  Using the separation of the
interferers, we argue that the number of interferers that can be
packed in the annulus at distance $i$ is bounded by a polynomial
depending on $i$ and the Assouad dimension of the
space.  

We argue that a general version of this `annulus argument' still holds
when directly used in fading decay spaces, after which we indicate
how other different variations carry over.

Recall the Riemann $\hat{\zeta}$-function,  $\hat{\zeta}(x) = \sum_{n \ge 1} n^{-x}$, which is known to converge for $x > 1$.
We build on a similar result in \cite{us:talg12} for metric spaces.


\begin{theorem}\label{annulusargument}
The fading parameter of a decay space $\calD = (V,f)$  with Assouad
dimension $A < 1$ and related constant $C$ is bounded by $\gamma = 
\gamma(r) \le C 2^{A+1}(\hat{\zeta}(2-A)-1)$. 
\end{theorem}

\begin{proof}
Let $R = r/2$.
Since $S$ is $r$-separated, the nodes in $S$ form an $R$-packing.
Since $\calD$ is doubling, there is a constant $C$ such that  for any $t > 0$, 
the maximal size of an $R$-packing in a ball of radius $tR$ centered around a point $x$ is, 
\begin{equation}
\mathcal{P}(B(x, tR),R) \leq C t ^A\ .
\end{equation}

We bound the received signal $I_S(x)$ at a listening node $x \in S$.
Let $g$ be a number.
Let $S_g = \{ y \in S' : f(y, x)< gR\}$ and let $T_g = S_g \setminus S_{g-1}$.
Then $S_2 = \emptyset$ since $S$ is $r$-spaced.

We first note that since $S_{g-1} \subseteq S_g$ and $S_2 = \emptyset$,
\[ \sum_{g \ge 3} \frac{|S_g \setminus S_{g-1}|}{g-1} 
  = \sum_{g \ge 3} \frac{|S_g|}{g-1} - \sum_{g \ge 2} \frac{|S_g|}{g}
  = \sum_{g \ge 3} |S_g| \left(\frac{1}{g-1} - \frac{1}{g}\right) 
  = \sum_{g \ge 3} \frac{|S_g|}{g(g-1)} \ . \]

Since each sender $y \in T_g$ is of distance at least $(g-1)R$ from $x$
the received signal from $y$ on $x$ is bounded by
$$I_y(x) = P/f(y, x)\leq \frac{P}{(g-1)R}\ \quad \forall y \in T_g\ .$$
Then,
\[ I_S(x) = \sum_{g \geq 3} I_{T_g}(x)
   \le \sum_{g \geq 3}|S_g \setminus S_{g-1}| \frac{P}{(g-1)R}
   \le \frac{P}{R} \sum_{g \geq 3} |S_g| \frac{1}{(g-1)^2}\ . \]
By the doubling property of $\calD$, the size of $S_g$ is 
  $$|S_g| \le \mathcal{P}(B(x,(g+1)R),R) \leq C(g+1)^A\ .$$
Thus, using that $g+1 \le 2(g-1)$, since $g\geq 3$,
  $$ \frac{|S_g|}{(g-1)^2} \leq \frac{C(g+1)^A}{(g-1)^2} = \frac{C 2^A}{(g-1)^{2-A}}\ .$$
Continuing,
  $$I_S(x) \leq \frac{P}{R} \sum_{g \geq 3} |S_g| \frac{1}{(g-1)^2} \leq \frac{2P}{r} \sum_{g \geq 3} \frac{C 2^A}{(g-1)^{2-A}} 
\leq \frac{2P}{r} C 2 ^A \left( \hat{\zeta}(2-A) -1 \right) =
\frac{\gamma(r) \cdot P}{r}\ ,$$
using the definitions of $R$ and $\gamma(r)$.
\end{proof}


\subsection{Common usage of the annulus argument}
We list some common types of lemmas in which the annulus argument is used and show how to use Theorem \ref{annulusargument} in the proofs for these lemmas.

A common usage of the annulus argument is to prove the following: if
$L$ is a set of links, using a uniform power assignment $P$, with
senders of a minimal mutual distance $r$ and with the longest link of
length at most a given constant times $r$, then $L$ forms a $q$-feasible set. 
For sets as described in Theorem \ref{annulusargument}, where all
\emph{nodes} are $r$-separated and a maximum link decay $f_{vv}$ at most
constant $r$, the transition is straightforward.
By the definition of affectance and Theorem \ref{annulusargument}, 
the affectance of $L$ on link $l_v$ with maximum decay $f_{vv}$ is at most 
$$ a_L(v) \leq \frac{I_L(v)}{P G_{vv}} \leq \frac{f_{vv} \cdot \gamma(r)}{r} \ ,$$
where $I_L(v) = \sum_{l_w \in L} 1/f_{wv}$.
To obtain a $q$-feasible set, we simply set $r = f_{vv}\gamma(r)/q$.

However, if only a separation on senders is defined (e.g. in \cite{us:talg12}), 
we use the triangular inequality to bound the interference at $r_v$ in terms of interference at $s_v$.
Requiring $f_{vv} < R$, we obtain $ I_L(r_v) \leq 2^\zeta I_L(s_v)$, since for any sender $s_x \in L$ by the triangle inequality
$$f(s_x, r_v)^{1/\zeta} \geq f(s_x, s_v)^{1/\zeta} - f(s_v, r_v) ^ {1/\zeta} \geq f(s_x, s_v)^{1/\zeta}/2\ ,$$
using that $f_{vv} < R \le f(s_x, s_v)^{1/\zeta}/2$.
And thus $f_{xv} \ge f(s_x,s_v)/2^\zeta$, so the argument holds as
before by adjusting $r$ with an extra $2^\zeta$ factor.
When $R \gg f_v$, the overhead factor is correspondingly smaller.

Examples of problems with centralized algorithms that use this form of annulus argument:
connectivity \cite{MoWa06,moscibroda06b,Moscibroda07,SODA12},
scheduling \cite{chafekar07,tonoyan2012},
flow-based throughput \cite{CKMPS08},
online capacity maximization \cite{fanghanel2013online},
and bounds on the utility of conflict graphs \cite{tonoyan2013,tonoyan2013a}.

For randomized algorithms, the annulus argument is used in 
a similar way to bound expected interference.  The expected
interference in a disk is bounded by arguments specific to the
analyzed algorithm.
 These arguments may or may not translate to the
decay space as discussed in Sec.~\ref{sec:theorytransfer}.
Instead of adjust the separation term $r$, thy typically adjust the transmission probabilities.
Once the expected interference in a disk is bounded, however, the
argumentation for bounding the total expected interference at a node
$x$, $E(I_S(x))$ follows Theorem \ref{annulusargument}.  


The probabilistic version of the annulus argument forms the core of
the analysis for many randomized distributed algorithms which often carry over without any significant further adjustments.
Example include (distributed) coloring \cite{YWHLa11}, local broadcast
\cite{Goussevskaia2008Local, Yu11, Yu12, FOMC12}, broadcast
\cite{DaumGKN13} and multiple-message broadcast \cite{YuHWTL12,
  YuHWYL13}, capacity \cite{pei2013distributed}, 
dominating set \cite{ScheidelerRS08} and (multihop) connectivity  \cite{PODC12, PODC13}, and dynamic packet scheduling \cite{pei2012low}.


\subsection{Beyond fading spaces}
Fading spaces do not completely characterize spaces with a
bounded fading parameter. One reason is that the definition of
doubling metrics is scale-invariant in that the packing
constraint holds for balls of any size, whereas we are often only
interested in balls of a fixed size (or in a limited range of sizes).

Consider, for instance, the metric space formed by a star centered at
node $x_0$ with $k$ leaves $x_1, x_2, \ldots, x_{k}$ at distance 
$k^2$ and one leaf $x_{-1}$ at distance $r$.  Suppose the
decay $f_{xy}$ equals the distance (so $\zeta=1$).  The doubling
dimension of this space is $k$, so unbounded.  Suppose also we are
interested in the separation term $r$, i.e., how well we can transmit
from $x_0$ to $x_{-1}$ in the presence of transmissions from the other
nodes. If $r=o(k)$, we find that the total interference at node
$x_{-1}$ is $\sum_{i=1}^k 1/k^2 = 1/k$, which is asymptotically
smaller than the signal received from $x_{0}$.

\section{Dependence on the Metricity in Approximations}
\label{sec:dependence}

With the pinpointing of the metricity parameter $\zeta$ as a key 
indicator of a decay space, the question
arises how it affects the complexity of fundamental problems. 
This differs from {\geomodel} where the path loss term
$\alpha$ has traditionally been viewed as a constant.


We explore here the approximability of the {\capacity} problem as a
function of innate properties of the decay space in question.
Given a set $L$ of links, the {\capacity} problem asks for 
maximum cardinality subset
of $L$ that is feasible.
The {\capacity} problem is fundamental, not only because it addresses
the basic question of how much wireless communication
can coexist, but also because it has been the underlying core routine
in other problems, including scheduling \cite{GHWW09}, throughput maximization (via
flow) \cite{wanwireless}, spectrum auctions \cite{hoeferspaa}, spectrum sharing \cite{us:Infocom12}, and connectivity and aggregation \cite{SODA12,PODC12}.

Our generic statement, Prop.~\ref{prop:metric}, along with known
approximation results \cite{SODA11,KesselheimESA12} in general metrics, 
implies that {\capacity} in decay spaces can be approximated
within a function of $\zeta$. Specifically, the
approximation of \cite{SODA11} (for monotone power) is exponential in $\zeta$, 
which was refined to $3^\zeta$ in \cite{us:ICDCS14}.

We can also observe that the known hardness construction for
``abstract SINR'' \cite{GHWW09} (see also \cite{SODA11}) implies that
$2^{\zeta(1-o(1))}$-approximation for {\capacity} is hard.  We include
the argument in the appendix for completeness.

\begin{theorem}
 \prob{Capacity} of equi-decay links is hard to approximate within $2^{\zeta(1-o(1))}$ factor.
This holds even if the algorithm is allowed 
arbitrary power control against an adversary that uses uniform power.
\label{thm:cap-hardness}
\end{theorem}

This leaves the question whether better results are possible in the
Euclidean metric and comparable decay spaces. 
Surprisingly, the answer depends on the exact definition of the
metricity parameter.
Specifically, {\capacity} with uniform power is then
approximable within a polynomial of $\zeta$, while for
a natural variant of the $\zeta$-parameter, exponential dependence
is still necessary.

\subsection{Improved Approximations in Bounded Growth Decay Spaces}
\label{sec:polya}

We show here that {\capacity} with uniform power can be approximated
within polynomial factors of $\zeta$ in Euclidean metrics.  More
generally, this holds for decay spaces of bounded growth, as we shall
define shortly.  Interestingly, it does not rely on the fading
behavior of the plane (i.e., that $\alpha > 2$). This appears to be
the first instance in the signal-strength literature where better
results are shown to be obtainable in the plane independent of
$\alpha$ than for general metrics.

The intuitive reason why uniform power in the plane proves to be easier
is two-fold.  The main cause for exponential dependence on $\zeta$
comes from the use of the triangular inequality. If one can ensure
that one angle is highly acute, the overhead of the inequality goes
down accordingly. In particular, the overhead in switching the
reference from a receiver to a sender of a link goes down if the
length of the link relative to the other distances is small.

We shall show that links with uniform power in bounded-growth decay spaces
satisfy a useful structural property that allows for improved
approximation for numerous problems.

\mypara{Bounded Growth Decay Spaces}
We shall consider decay spaces that have upper bounds on two measures
that restrict growth: 
the doubling dimension (from Sec.~\ref{sec:fading}), 
and the independence dimension, defined in \cite{GHWW09} for metrics
and adapted as follows to decay spaces.

\begin{definition}[\cite{GHWW09}]
A set $I$ of points in a decay space $\calD=(V,f)$ is \emph{independent} w.r.t.\ a point $x \in V$ if $B(z, f_{zx})\cap I = \{ x \}$ for each $z \in I$.
The \emph{independence dimension} of $\calD$ is the size of the largest independent point set.
\end{definition}

Spaces of bounded independence dimension $\inddim$ have the following
useful property: for any point $x \in V$, there is a set $J_x \subset
V$ of at most $\inddim$ points that \emph{guard} $x$ in the following
sense: $\min_{y \in J_x} d(z,y) \le d(z,x)$, for any point $z \in V \setminus \{x\}$. A node $y$ \emph{guards} node $x$ \emph{from} node $z$ if
$d(z,y) \le d(z,y)$.

Welzl \cite{Welzl08} has made a number of useful observations of
metrics of bounded independence dimension. He showed that the number of guards needed in a
metric is indeed exactly its independence dimension. In a Euclidean
space $\mathbb{R}_n$, it equals the maximum number of unit vectors
that form pairwise angles of more than $60\degree$.
Therefore, the independence is at most the so-called kissing number,
the maximum number of disjoint open balls of radius 1 that can touch
the unit ball. This number grows exponentially in the dimensions but
its exact value is not known for most dimensions.

As a simple example, let us see how six guards suffice in the plane.
Given a point $x$, divide the plane into six $60\degree$ sectors around $x$
and partition $V$ accordingly into sets $S_1, S_2, \ldots, S_6$.
Let $J_x$ consist of the nearest point to $x$ in each of the six sectors.
The guarding property follows from the fact that the angle $\angle g_i x y_i$ is at least $60\degree$, for each point $y_i \in S_i$ and guard $g_i \in J_x$.

We define a decay space to be \emph{bounded-growth} if it has bounded
independence dimension and its quasi-distance metric has a bounded
doubling dimension.  (The dimension of a decay space and its
quasi-distance metric is the same.)

The doubling and independence dimensions are actually incomparable.
The uniform metric, where all decays equal 1, is of independence
dimension 1 but unbounded doubling dimension.
The following curious construction of Welzl \cite{Welzl08} gives a
metric of doubling dimension 1 whose independence dimension is unbounded:
Let $V=\{v_{-1}, v_0, v_1, \ldots, v_n\}$ with
$d(v_{-1},v_i)=2^i-\epsilon$, for $0 < \epsilon \le 1/4$, and $d(v_j, v_i)
= 2^i$, for $i, j \ne -1, j < i$. We leave it to the curious reader to
verify that any ball (only those of radius $2^i$ or $2^{i}-\epsilon$
matter) can be covered with two balls of half the radius and that $V
\setminus \{v_{-1}\}$ are independent with respect to $v_{-1}$.

\mypara{Amicability}
The following definition originates in \cite{infocom11} and 
was formally stated in \cite{dams2013sleeping} as \emph{$C$-independent} conflict graphs.

\begin{definition}
A set $L$ of links is \define{$h(\zeta)$-amicable} if there is a constant $c$ such that, for
any feasible subset $S \subseteq L$, there is a subset $S' \subseteq S$ with $|S'| \ge c |S|/h(\zeta)$ such that for any vertex $v \in L$, $a_v(S') \le c$ (using uniform power).
\end{definition}

It is known that sets in {\geomodel} in metric spaces are $2^{O(\alpha)}$-amicable \cite{infocom11}.

Various decentralized capacity-type problems with uniform power have been treated with no-regret minimization techniques, relying only on the amicability property of the instances. This started with a distributed constant approximation for {\capacity} \cite{Dinitz2010,infocom11}, and was extended to deal with jamming \cite{dams2013jamming}, online requests against stochastic adversaries \cite{GHKSV13}, and changing spectrum availability \cite{dams2013sleeping}.
Our $\alpha^{O(1)}$-bound on amicability improves these results in 
the bounded-growth metrics.

We show that growth-bounded instances are $\zeta^{O(1)}$-amicable,
thus obtaining improved approximations for the above problems (as
functions of $\zeta$).

\mypara{Capacity approximation via bounds on amicability}

To bound amicability, we first show how to turn feasible
sets in doubling spaces into well separated sets at 
limited cost. The proof is deferred to the appendix.

\begin{lemma}
  Let $S$ be a feasible set of links in a decay space whose
  quasi-distance metric has doubling dimension $A'$.  Then, $S$ can be
  partitioned into $O(\zeta^{2A'})$ sets, all of which are
  $\zeta$-separated.
\label{lem:planar_separation}
\end{lemma}

We are now ready to prove the structural result of this section.

\begin{theorem}
  Let $L$ be a set of links in a decay space of independence dimension
  $\inddim$ and whose quasi-distance metric has doubling dimension
  $A'$.  Then, $L$ is $O(\inddim \zeta^{2A'})$-amicable.
\label{thm:indep}
\end{theorem}

\begin{proof}
Let $S \subseteq L$ be any feasible subset of $L$.
By Lemma \ref{lem:planar_separation}, there is a subset
$\hat{S} \subseteq S$ of size $\Omega(|S|/\zeta^{2A'})$ that is 
$\zeta$-separated.
Let $S' = \{l_v : a_v(\hat{S}) \le 2 \}$ be the subset of links in
$\hat{S}$ with low out-affectance.
Note that $\sum_{l_v \in \hat{S}} a_v(\hat{S}) = 
\sum_{l_v \in \hat{S}} a_{\hat{S}}(v) \le |\hat{S}|$, by feasibility,
so the average
out-affectance of links in $\hat{S}$ is at most 1, and at least half
the links will have at most double the out-affectance.
Thus, 
 \[ |S'| \ge |\hat{S}|/2 = \Omega(|S|/\zeta^{2A'})\ . \]

Consider any link $l_v \in L$.  Let $J_v = \{g_1, g_2, \ldots, g_t\} $
be the indices of senders in $|S'|$ that guard the sender $s_v$ of
$l_v$, where $t \le \inddim$. Partition $S'$ into sets 
$S_1, S_2, \ldots, S_t$, where $s_{g_i}$ is contained in $S_i$ and guards
$s_v$ from the senders of other links in $S_i$.
Consider any set $S_i$ and let $l_x$ be a link in $S_i$.
Since $s_{g_i}$ guards $s_v$ from $s_x$, $d(s_{g_i},s_x) \le d(s_v, s_x)$.
Then, additionally using the triangular inequality and that $S_i$ is $\zeta$-separated,
\[ d(s_{g_i},s_x) \le d(s_v,s_x) \le d_{vx} + d_{xx} \le (1 + 1/\zeta)d_{vx}\ . \]
So, 
$f(s_{g_i},s_x) = d(s_{g_i},s_x)^\zeta 
  \le (1 + 1/\zeta)^\zeta f_{vx}  \le e \cdot f_{vx}$.
In a similar way, we obtain that
  $d_{g_i x} \le d(s_{g_i},s_x)  + d_{xx} \le (1 + 1/\zeta)d(s_{g_i},s_x)$,
so 
 \[ f_{g_i x} \le (1 + 1/\zeta)^\zeta f(s_{g_i},s_x) \le e \cdot f(s_{g_i},s_x)\ . \]
Combining, we get that 
$f_{g_i x} \le e \cdot f(s_{g_i},s_x) \le e^2 f_{v x}$.
We can then bound the out-affectance of $l_v$ on $S_i$ by
\[ a_v(S_i) = \sum_{l_x \in S_i} a_v(x)
  = \sum_{l_x \in S_i} c_x \cdot \frac{f_{xx}}{f_{vx}}
  \le a_v(g_i) 
    + \sum_{l_x \in S_i \setminus\{l_{g_i}\}} c_x \cdot \frac{e^2 \cdot f_{xx}}{f_{g_i x}}
  = 1 + e^2 \cdot a_{g_i}(S_i) \le 1 + 2e^2\ , \]
using the definition of $S'$ in the last inequality.
Then, $a_v(S') \le (1+2e^2) \inddim$.
Then, $L$ satisfies the definition of amicability with $h(\zeta) =
O(D \zeta^{2A'})$ and $c = (1+2e^2) \inddim$.
\end{proof}

We arrive at the main result of this section, whose proof is given in the appendix. Algorithm \ref{alg:capfixtri} combines the characteristics of the capacity algorithms of \cite{GHWW09} and \cite{SODA11}.

\begin{algorithm}[h]
\caption{Capacity for uniform power in bounded-growth decay spaces.}\label{alg:capfixtri}
\begin{algorithmic}
\STATE Let $L$ be a set of links using uniform power and let $X \leftarrow \emptyset$
\FOR {$l_v \in L$ in order of increasing $f_{vv}$ value}
\IF {$l_v$ is $\zeta/2$-separated from $X$ and $a_v(X) + a_X(v) \leq 1/2$} \label{alg:tri1/2}
\STATE $X \leftarrow X \cup \{l_v\}$
\ENDIF
\ENDFOR
\STATE Return $S \leftarrow \{l_v \in X| a_X(v) \leq 1\}$
\end{algorithmic}
\end{algorithm}

\begin{theorem}
  Uniform power {\capacity} $\zeta^{O(1)}$-approximable in
  bounded-growth decay spaces (by Algorithm \ref{alg:capfixtri}).
  In particular, it is $O(\alpha^4)$-approximable on the plane, for any $\alpha$.
\label{thm:cap-poly-bndgwth}
\end{theorem}

This is actually the first SINR approximation result (for capacity or related problems) that is sub-exponential in $\alpha$.

\subsection{Inapproximability results for a variant of metricity}
\label{sec:expon-phi}

\mypara{Metricity variant  $\boldsymbol{\varphi}$} 
Alternative measures of the metric-like behavior of a space
$\calD=(V,f)$ can
be concocted. A particularly natural one is the 
parameter $\varphi$ that bounds the \emph{multiplicative} factor within which 
$f$ satisfies a relaxed triangular inequality:
\[ \varphi = \max_{x,y,z \in V} \frac{f_{xy} + f_{yz}}{f_{xz}}\ . \]
So, $\varphi$ is the smallest value such that $f_{xz}\le
\varphi(f_{xy}+f_{yz})$, for every $x,y,z\in V$.
For comparison with $\zeta$, we define $\phi = \lg \varphi$.

Examining the proofs of the various results for {\capacity} 
and \emph{inductive independence} \cite{hoeferspaa}, we find that the triangular
inequality is applied to compare lengths that are within constant
factor of each other, in which case the overhead is comparable to the case of $\zeta$. Thus, the results hold also in terms of $\phi$.

\begin{observation} {\capacity}, both with monotone power
  \cite{SODA11,us:ICDCS14} and arbitrary power control
  \cite{KesselheimSODA11}, is approximable within $2^{O(\phi)}$.
  Other results with effective (exponential) approximations in terms
  of similar bounds hold for inductive independence
  \cite{hoeferspaa,us:SODA13} and relationships between power control
  and monotone power \cite{us:SODA13}.
\end{observation}

Bounds on inductive independence also have numerous implications,
including connectivity and aggregation \cite{SODA12,PODC12}, spectrum
auctions \cite{hoeferspaa,HoeferK12}, dynamic packet scheduling
\cite{sirocco12,kesselheimStability}, and distributed scheduling
\cite{KV10,icalp11}.

We can observe that $\zeta \le \phi$. Namely, for any nodes $x, y, z$,
$f_{xz}^{1/\zeta} \le f_{xy}^{1/\zeta} + f_{yz}^{1/\zeta} 
  \le 2 \max(f_{xy}^{1/\zeta},f_{yz}^{1/\zeta}) 
  = 2 (\max(f_{xy},f_{yz}))^{1/\zeta} \le (f_{xy}+f_{yz})^{1/\zeta}$, 
using the definition of $\zeta$.
Thus, $f_{uv} \le 2^{\zeta}(f_{uw} +f_{wv})$.
Hence, lower bounds in terms of $\zeta$
carry over to lower bounds in terms of $\phi = \log \varphi$,
so exponential approximations in terms of $\phi$ are best
possible in general metrics.

A converse relation between $\zeta$ and $\phi$ does not exist, however.
Consider the instance on three points $V=\{a,b,c\}$ with 
$f_{ab}=1$, $f_{bc}=q$ and $f_{ac} = 2q$. 
Then, one can verify that $\phi \le 2$, while 
$\zeta = \theta(\log q/\log\log q)$, which is unbounded.

We find that {\capacity} in
bounded-growth spaces is still exponentially hard in terms of $\phi$.
We give a construction that is embedded on a pair of lines,
that holds for arbitrary values of a parameter $\alpha$.
For decays within the lines, it uses the usual distance function raised to power $\alpha$,
while between the lines, it uses two fixed decays: $n^\alpha$ and $n^{\alpha+1}$.
It then also shows that strong hardness holds even when none of the decay
functions are particularly fast growing. 
The proof is deferred to the appendix.

\begin{theorem}[\cite{GHWW09}]
\prob{Capacity} of equi-decay links in bounded-growth decay spaces
is hard to approximate within $2^{\phi(1-o(1)}$-factor.
This holds even if the algorithm is allowed 
arbitrary power control against an adversary that uses uniform power.
\label{thm:hardness}
\end{theorem}

We note that the decays used in the construction were all 
in the range $d^{\alpha'}$ and $d^{\alpha'+1}$ between pairs of distance $d$.
%
This result thus shows that huge decays (or, path loss) are not
needed \emph{per se} to get large approximation hardness. Rather, it is the
differences in decay among spatially related points that is the
cause. 

\newpage

\bibliographystyle{abbrv}
\bibliography{references}

\appendix

\section{Missing proof from Section \ref{sec:dependence}}

\noindent \textbf{Theorem \ref{thm:cap-hardness}.} \emph{
 \prob{Capacity} of equi-decay links is hard to approximate within $2^{\zeta(1-o(1))}$ factor.
This holds even if the algorithm is allowed 
arbitrary power control against an adversary that uses uniform power.
}

\begin{proof}
Given a graph $G=(V,E)$, form a set $L$ of links of unit-decay with a
link $l_i$ for each node $v_i$ and with the (bi-directional) decay of
$f_{ij} = f_{ji}$ as 2 if $v_iv_j \in E$ and $1/n$ if $v_iv_j \not\in
E$.

If $S$ is a feasible set of links in $L$, then it contains no two
links $l_i$ and $l_j$ that form an edge in $E$, no matter what
power they assume.  Similarly, if $I$ is an independent set in $G$,
then if $S_I$ is the corresponding set of links, the affectance of any
given link $l_i$ in $S_I$ when using uniform power is at most $(n-1) \cdot
1/n < 1$; thus, $S_I$ is feasible.  Hence, there is a one-one
correspondence between independent sets in $G$ and feasible sets in
$L$, as well as between sets that are feasible and those that are feasible
under uniform power.

Now, observe that $\zeta \le \lg n$, as $n$ is the maximum ratio
between decays, and the bound is actually tight.  The
$n^{1-o(1)}$-approximation hardness of \prob{Max Independent Set}
\cite{khot2006better} then translates to $|L|^{1-o(1)} =
2^{\zeta(1-o(1)}$-approximation hardness for {\capacity}.
\end{proof}

\section{Missing proofs from Section \ref{sec:polya}}

We shall make use of the following technique.

\begin{lemma}[Signal-strengthening \cite{HW09}]
There is a polynomial-time algorithm that, for any given $p, q$, 
partitions any $p$-feasible set into 
$\lceil 2q/p \rceil^2$ sets, all $q$-feasible.
%
\label{lem:signal-strength}
\end{lemma}

We first argue that feasible sets under uniform power must be somewhat separated (or, $1/\zeta$-separated), independent of metric.

\begin{lemma}
Let $S$ be an $e^2/\beta$-feasible set of links under uniform power and assume
$\zeta \ge 1$.
Then, $S$ is $1/\zeta$-separated.
\label{lem:onezetasep}
\end{lemma}

\begin{proof}
Suppose otherwise. Then, there are two links $l_v,l_w$ in 
$S$ that are not $1/\zeta$-separated.
There are three cases, depending on which pairwise distance bound is violated.

Consider first the case when $d(s_v,r_w) < (1/\zeta)
\max(d_{vv},d_{ww})$.  Since the two links are feasible
simultaneously, the signal received by $r_w$ from $s_w$ is at least as
strong as that from the other sender $s_v$ (since $\beta \ge 1$).
So, $d_{ww} \le d(s_v,r_w)$, implying that 
$d(s_v,r_w) < (1/\zeta) d_{vv}$.
Then, by the triangular inequality and these bounds, 
\[ d(s_w,r_v) \le d_{ww} + d(r_w,s_v) + d_{vv} \le 2d(s_v,r_w)+d_{vv}
    < (1 + 2/\zeta)d_{vv}\ . \]
Thus, $f_{wv} < (1+2/\zeta)^\zeta f_{vv} \le e^2 f_{vv}$.
It follows that
 \[ a_w(v) = c_v \frac{f_{vv}}{f_{wv}} \ge \frac{c_v}{e^2} > \frac{\beta}{e^2}\ . \]
This contradicts the assumption that $l_v$ and $l_w$ coexist in the same $e^2$-feasible set.

Consider next the case when $d(r_v,r_w) < (1/\zeta) \max(d_{vv},d_{ww})$.
Without loss of generality, assume  $d(r_v,r_w) < d_{vv}/\zeta$.
By the triangular inequality,
$d_{vw} \le d_{vv} + d(r_v,r_w) < d_{vv}(1 + 1/\zeta)$, 
implying that $f_{vw} < (1+1/\zeta)^\zeta f_{vv} \le e \cdot f_{vv}$,
leading to a contradiction as before.
Finally, the case when $d(s_v,s_w) < \max(d_{vv},d_{ww})$ is symmetric to the
previous one when swapping senders and receivers.
Hence, the claim.
\end{proof}

We next show that in doubling metrics, the separation factor can be
expanded by a polynomial factor at the cost of a polynomial factor.

\begin{lemma}
Let $\tau$ and $\eta$ be positive parameters, $\tau < \eta$.
Let $S$ be a $\tau$-separated set of links 
in a decay space whose quasi-distance metric has doubling dimension $A'$.
Then, $S$ can be partitioned into $O((\eta/\tau)^{A'})$ sets each of which 
is $\eta$-separated.
\label{lem:sep-strength}
\end{lemma}

\begin{proof}
Consider a link $l_v$ in $S$. Let $R_v$ be the set
of links in $S$ whose receivers are within distance $\eta \cdot
d_{vv}$ from $r_v$. Then, we have a set of $|R_v|$
disjoint balls of radius $\tau d_{vv}/2$ that are properly contained
in a ball of radius of $(\eta +\tau/2) d_{vv}$ (around $r_v$).
By the definition of the Assouad dimension, 
\begin{equation}
|R_v| \le C \left(\frac{\eta+\tau/2}{\tau/2}\right)^{A'}
   = C \left((2 \eta+1)/\tau\right)^{A'} \ .
\label{eq:rv}
\end{equation}

We now form the graph $G_{S} = (V,E)$, where $V = S$ and 
$(l_v, l_w) \in E$ iff $l_v \in R_w$ or $l_w \in R_v$.
Let $\rho = \max_{l_v \in S} |R_v| \le C ((2\eta+1)/\tau)^{A'} 
  = O((\eta/\tau)^{A'})$.
Form a total order $\prec$ on the nodes by non-increasing link length.
By (\ref{eq:rv}), each node has at most $\rho$ neighbors that follow it in the ordering (because if $l_v \prec l_w$ then $l_w \in R_v$). That is, $\prec$ is a 
\emph{$\rho$-inductive} (or, \emph{$\rho$-degenerate}) ordering of $G$.
Coloring the graph first-fit according to $\prec$ then uses at most $\rho+1$ colors. 
To complete the proof, we observe that a set of links is $\eta$-separated
if and only if the corresponding set of vertices in the graph is independent 
(graph-theoretically). 
\end{proof}

Put together, we obtain a sparsity-strengthening lemma in doubling spaces.

\noindent \textbf{Lemma \ref{lem:planar_separation}.} \emph{
  Let $S$ be a feasible set of links in a decay space whose
  quasi-distance metric has doubling dimension $A'$.  Then, $S$ can be
  partitioned into $O(\zeta^{2A'})$ sets, all of which are
  $\zeta$-separated.
}

\begin{proof}
Recall that by the signal strengthening Lemma
\ref{lem:signal-strength}, $S$ can be partitioned into at most
$(e^2/\beta+1)^2$ sets each of which is $e^2/\beta$-feasible.  Let
$S'$ be such a set.  By Lemma \ref{lem:onezetasep}, $S'$ is
$1/\zeta$-separated, so by Lemma \ref{lem:sep-strength}, $S'$ can be
partitioned into $O(\zeta^{2A'})$ sets, each of which is
$\zeta$-separated.
\end{proof}

\noindent \textbf{Theorem \ref{thm:cap-poly-bndgwth}.} \emph{
  Uniform power {\capacity} is $\zeta^{O(1)}$-approximable in
  bounded-growth decay spaces.
}

\begin{proof}
We use Algorithm \ref{alg:capfixtri}.

Let $L$ be a set of links and $S$ and $X$ be the sets computed by the
algorithm on input $L$. 
Let $\prec$ denote the order in which the algorithm processes the links.
Note that by rearrangement and the construction of
$X$, $\sum_{l_v \in X} a_X(v) \leq |X|/2$.  Thus, the average
in-affectance of a node is at most $1/2$, and by Markov's inequality
\begin{equation}
|S| \geq 1/2\cdot |X| \ .
\label{eq:s-x}
\end{equation}
Let $OPT$ be a maximum capacity subset of $L$. Let $OPT' \subseteq
OPT$ be the subset of $OPT$ promised by Thm.~\ref{thm:indep} that has
cardinality $\Omega(|OPT|/\zeta^{2A})$ and satisfies $a_v(OPT') \le
C$, for every $l_v \in L$.  Observe that the proof of
Thm.~\ref{thm:indep} actually ensures that $OPT'$ is
$\zeta$-separated.

Let $Z = OPT' \setminus X$.  Partition $Z$ into $Z_1$ and $Z_2$, where
links in $Z_1$ failed the requirement of $\zeta/2$-separability from
$X$, while those in $Z_2$ passed the separability requirement but failed
the affectance test. We proceed to bound $|Z_1|$ and $|Z_2|$ in terms of $|X|$.

First, observe that for each link $l_v$ in $X$, at most one link in
$Z_1$ can fail to be $\zeta/2$-separated from $l_v$, as otherwise
$Z_1$ would not be $\zeta$-separated.  That implies that $|Z_1| \le
|X|$.

Now, let $l_w$ be a link in $Z_2$ and let $X_w = \{l_v \in X : l_v
\prec l_w\}$ be the links in $X$ that precede $l_w$ in the decay
order.  Let $l_u$ be a link in $X_w$.
Then, $f_{uu} \le f_{ww}$, $d_{uu} \le d_{ww}$ and $c_u \le c_w$,
since $l_u \prec l_w$.
Using the triangular inequality, the fact that $d_{uu} \le d_{ww}$,
and that $l_w$ is $\zeta/2$-separated from $X_w$, we get that
\[ d_{uw} \le d_{uu} + d_{wu} + d_{ww} \le d_{wu} + 2 d_{ww} 
   \le (1 + 4/\zeta) d_{wu}\ . \]
Thus, $f_{uw} \le (1 + 4/\zeta)^\zeta f_{wu} \le e^4 \cdot f_{wu}$.
Hence, since $f_{uu} \le f_{ww}$ and $c_u \le c_w$, 
\[ a_w(u) = c_u \frac{P/f_{wu}}{P/f_{uu}} 
          = c_u \frac{f_{uu}}{f_{wu}} \le e^4 c_w \frac{f_{ww}}{f_{uw}} = e^4 \cdot a_u(w)\ . \]
Thus, $a_w(X_w) \le e^4 \cdot a_{X_w}(w)$.
By definition of $Z_2$, $a_w(X_w) + a_{X_w}(w) \ge 1/2$.  
Combining the last two inequalities, we get that 
$a_X(w) \ge a_{X_w}(w) \ge 1/(2e^4+1)$.
Summing this inequality over links in $Z_2$,
\begin{equation}
\sum_{l_v \in X} \sum_{l_w \in Z_2} a_v(w) = \sum_{l_w \in Z_2} a_X(w) \ge \frac{|Z_2|}{2e^4+1}\ . 
\label{eq:sum-axw}
\end{equation}
On the other hand, by amicability,
\begin{equation}
\sum_{l_v \in X} \sum_{l_w \in Z_2} a_v(w) = \sum_{l_v \in X} a_v(Z_2) \le C \cdot |X|\ .
\label{eq:sum-amic}
\end{equation}
Combining (\ref{eq:sum-axw}) and (\ref{eq:sum-amic}),
we obtain that
\[ |Z_2| \le (2 e^4 +1)C \cdot |X|. \]
Thus, $|Z| = |Z_1| + |Z_2| \le ((2 e^4+1) C +1) \cdot |X|$,
and
\[ |OPT'| = |Z| + |X \cap OPT'| \le ((2 e^4+1) C +2) \cdot |X|
 \le (4 e^4 C +2C+4) \cdot |S| \ , \]
using (\ref{eq:s-x}). Hence, 
$|OPT| = O(\zeta^{2A})|OPT'| = O(\zeta^{2A}|S|)$,
as claimed.
\end{proof}

\section{Missing proof from Section \ref{sec:expon-phi}}


\noindent\textbf{Theorem \ref{thm:hardness}}(\cite{GHWW09}).\emph{
\prob{Capacity} of equi-decay links in bounded-growth decay spaces
is hard to approximate within $2^{\phi(1-o(1)}$-factor.
This holds even if the algorithm is allowed 
arbitrary power control against an adversary that uses uniform power.
} 

\begin{proof}
By reduction from the maximum independent set problem in graphs.
Let $\alpha$ be arbitrary value satisfying $\alpha \ge 1$, 
denoting the maximum path loss term and let $\alpha' = \alpha-1$.
Assume for simplicity that $N=0$ and $\beta=1$.
Let $d_2(\cdot)$ refers to the standard Euclidean distance.

Given graph $G(V, E)$, form a set $L$ of links
with link $l_i=(s_i,r_i)$ for each vertex
$v_i\in V$ located in the plane.
The senders are located on the vertical line segment $[(0,0),(0,n)]$
and the receivers on the segment $[(n,0),(n,n)]$:
$s_i$ at point $(0,i)$ and $r_i$ at point $(n,i)$.

Decays between points on the same line (both senders or both receivers)
are set to their distance to the power of $\alpha'$.
For decays between points on different lines, 
we use two fixed decays: $n^{\alpha'}$ and $n^{\alpha'+1}$.

Formally, for links $l_i$ and $l_j$, let
 \[ f_{ij} = f(s_i,r_j) = \begin{cases}
      d_2(s_i,r_j)^{\alpha'} = n^{\alpha'} & \mbox{if } i=j \\
      n^{\alpha'} - \delta & \mbox{if } v_iv_j \in E \\
      n^{\alpha'+1} & \mbox{if } v_iv_j \not\in E\ ,
   \end{cases} \]
where $0 < \delta < 1/2$.
Also, let $f(s_i, s_j) = f(r_i, r_j) = d_2(s_i, s_j)^{\alpha'} = |i-j|^{\alpha'}$.

With uniform power $P$,
we have that for each $i \ne j$,
\[ a_i(j) = \frac{P/f_{ij}}{P/f_{j}} = \frac{n^{\alpha'}}{f_{ij}}
  \begin{cases}
    \,\,  > 1 & \mbox{if } v_iv_j \in E \\
    \,\, \le 1/n & \mbox{if } v_iv_j \not\in E.
  \end{cases} \]
Hence, a set $S \subset L$ of links is feasible iff $V_S = \{v_i \in V
: l_i \in S\}$ is an independent set.

For the case of power control, consider a pair of links $l_v, l_w$
and let $\calP$ be any power assignment on the links.
If $(v,w) \in E$, then 
$f_{vw} \cdot f_{wv} = (n^{\alpha'}-\delta)^2$, which implies that
\[ a_v^\calP(w) \cdot a_w^\calP(v) 
    \ge \beta^2 \frac{f_{vv} \cdot f_{ww}}{f_{vw} \cdot f_{wv}}
    = \beta^2 \frac{n^{2\alpha'}}{(n^{\alpha'}-\delta)^2} > \beta^2 = 1\ . \]
So, at least one of $a^\calP_v(w)$ and $a^\calP_w(v)$ must be
greater than one, implying that no power assignment allows
$l_v$ and $l_w$ to be simultaneously feasible.
Hence, any feasible set $S$ must correspond to an independent set in $G$,
and we know that any independent set in $G$ can be made feasible in
$L$ using uniform power. Solutions to {\capacity} on $L$ are therefore
in one-one correspondence with solutions to \prob{Max Independent Set}
on $G$, preserving solution size.

Regarding $\varphi$, observe that $f(s_i,s_j) = f(r_i,r_j) \ge 1$.
Then, we can verify that for any triplet $a,b,c$ of points used in $L$,
\[ f_{ac} \le 2 n \max(f_{ab},f_{bc}) \ . \]
Thus, $\varphi = O(n)$.
Hence, if \prob{Capacity} is approximable within $f(\varphi)$ factor,
then \prob{Max Independent Set} is approximable within $O(f(n))$ factor.
In particular, the $\Omega(n^{1-o(1)})$-computational hardness of \prob{Max Independent
Set} \cite{khot2006better} implies equivalent $\Omega(\varphi^{1-o(1)})$-hardness for \prob{Capacity}.

Finally, we examine the bounded-growth properties of the space. 
A $t$-ball with $t < n^{\alpha'}-\delta$ contains either only senders or
only receivers, and such sets can be covered by two balls of half the radius.
However, any subset of nodes can be covered with four balls of radius
at least $(n^{\alpha'}-\delta)/2$, two on each line. Thus, the decay
space is doubling (with $A \le \lg 4 = 2$).
As for independence, all nodes on a line are closer to each other than
they are to any node on the other line. Thus, an independent set with
respect to a point $x$ contains at most two points from the same line
as $x$ and at most one point from the other line, for an independence
dimension of 3.
\end{proof}

\end{document}